\documentclass[11pt,a4paper]{amsart}

\newtheorem{theorem}{Theorem}[section]
\newtheorem{corollary}[theorem]{Corollary}
\newtheorem{lemma}[theorem]{Lemma}
\newtheorem{remak}[theorem]{Remark}

\theoremstyle{definition}
\newtheorem{definition}[theorem]{Definition}

\newcommand{\Z}{\mathbb{Z}}

\newcommand{\B}{\mathcal{B}}
\newcommand{\R}{\mathbb{R}}
\newcommand{\N}{\mathbb{N}}

\newcommand{\disp}{\displaystyle}

\date{}

\title[Preservation of Hausdorff spectrum]{On the spectral Hausdorff dimension of 1D discrete Schr\"odinger operators under power decaying perturbations}
      
\author{V. R. Bazao}\address{Departamento de Matem\'{a}tica -- UFSCar, S\~{a}o Carlos, SP, 13560-970 Brazil} \author{S. L. Carvalho}\address{Departamento de Matem\'atica, UFMG, Belo Horizonte, MG,  30161-970 Brazil}  \author{C. R. de Oliveira} \address{Departamento de Matem\'{a}tica -- UFSCar, S\~{a}o Carlos, SP, 13560-970 Brazil}

\begin{document}
\maketitle

\begin{abstract}
 We show that spectral Hausdorff dimensional properties of discrete Schr\"odinger operators with (1)~Sturmian potentials of bounded density and (2)~a class of sparse potentials are preserved under suitable  polynomial decaying perturbations, when the spectrum of these perturbed operators have some singular continuous component.
 \end{abstract}

\

\thispagestyle{empty}

\section{Introduction}\label{intro}

We present results about preservation of spectral Hausdorff dimensional properties for some discrete Schr\"odinger operators~$H$, with (real) potentials $V=\left\{V(n)\right\}$,  on $l^2(\Z)$ or $l^{2}(\N)$,  of the form
\begin{equation}\label{H} 
(H\psi)(n) = \psi(n+1) + \psi(n-1) + V(n)\psi(n),
\end{equation}
under suitable power decaying (real) perturbations $P=\{P(n)\}$, that is, when~$V$ is replaced with $V+P$. Here,  the term  spectral measure of~\eqref{H} acting in~$l^2(\mathbb N)$ refers to the measure associated with the cyclic vector~$\delta_{1}$, where $\delta_j=(\delta_{ij})_{i\ge 1}$; in case~\eqref{H} acts in~$l^{2}(\Z)$, then the terminology refers to the spectral measures associated with both~$\delta_{0}$ and~$\delta_{1}$. On the half-line~$\mathbb N$, each self-adjoint realization of~$H$ in~\eqref{H} is given by a phase boundary condition
\begin{equation}\label{boucon}
\psi(0)\cos\varphi+\psi(1)\sin\varphi=0,\quad \varphi\in(-\pi/2,\pi/2],
\end{equation} which will be denoted by~$H_\varphi$.

Denote by $\sigma(T)$ the spectrum of a self-adjoint operator~$T$, and by~$\sigma_{\mathrm{p}}(T)$, $\sigma_{\mathrm{sc}}(T)$ its pure point and singular continuous spectra, respectively; if~$\mu$ is a Borel measure on~$\mathbb R$, we say that~$\mu$ is supported on the Borelian~$S$ if $\mu(\mathbb R\setminus S)=0$.
 
We are particularly interested in the family~$\{H_{\lambda,\theta,\rho}\}$ of the so-called Sturmian operators, that is, 
the family of operators~\eqref{H} with almost periodic Sturmian potentials 
\[ 
V(n)=V_{\lambda,\theta,\rho}(n)=\lambda \chi_{[1-\theta,1)}(n\theta+\rho \ \textrm{mod}\,1),\quad n\in\Z,
\]
 where~$0 \neq \lambda \in \R$ is the coupling constant,~$\theta \in [0,1)$ is an irrational rotation number of bounded density (in Section~\ref{applic} this notion is recalled) and~$\rho \in [0,1)$ is the phase. It is well known~\cite{D,DKL,JL2} that each~$H_{\lambda,\theta,\rho}$  has purely~$\alpha$-Hausdorff continuous spectrum (and that~$\sigma(H_{\lambda,\theta,\rho})$ has zero Lebesgue measure) for some $\alpha\in (0,1)$. Here, $\alpha = \frac{2\gamma_1}{\gamma_1 + \gamma_2}$, where $\gamma_1(\theta,\lambda),\gamma_2(\theta,\lambda)>0$ are such that
 \[
 C_{1}L^{\gamma_{1}}\leq \|u\|_{L}\leq C_{2}L^{\gamma_{2}},
 \] for positive constants $C_1,C_2$ and every solution~$u$ to the eigenvalue equation $H_{\lambda,\theta,\rho}u=Eu$ with normalized initial conditions
\begin{equation}\label{NIC}
|u(0)|^2+|u(1)|^2=1;
\end{equation} $\|u\|_L$ is the truncated $l^2([1,L])$ norm (see details ahead). It is precisely in the proof of such inequalities that the bounded density hypothesis plays an important role. 
 
Since the proof of this~$\alpha$-Hausdorff continuity property relies heavily on the particular structure of Sturmian potentials, it is an interesting question whether it is preserved under certain perturbations. However, since the spectra of the (bounded density) Sturmian operators are purely singular continuous of zero Lebesgue measure, by considering a rank one perturbation $a\delta_1$ with intensity $a\in\mathbb R$, it follows from Simon-Wolff's criterion~\cite{SW} that the spectrum of the perturbed operator $H_{\lambda,\theta,\rho}+ a\delta_1$ is pure point for  (Lebesgue) a.e.~$a$,  whereas for~$a$ in a generic set (i.e., Baire typical) of intensities, the spectrum of the perturbed operator has a singular continuous~\cite{dRMS,gordon} component.  Thus, the following stability result for suitable decaying perturbations, namely, a preservation of the~$\alpha$-Hausdorff continuity of spectral measures, only applies when the perturbed Sturmian operator has a singular continuous component. 

Since there is a lack of results on preservation of (nontrivial, i.e., different form zero and one) Hausdorff dimensional properties under perturbations, we underline that the results below can be considered interesting even in cases they apply  for parameters  (Liouville-like, say) in a set of zero Lebesgue measure.

\begin{theorem}\label{teorema1.1}
Let $\theta$ be a bounded density irrational number and $\gamma_1,\gamma_2,\alpha$ as above. Then, for every $\rho\in [0,1)$ and $\lambda\neq 0,$ any singular continuous component of the spectral measure associated with the operator
\begin{equation} \label{P.S.}
(H^{P}_{\lambda,\theta,\rho}\psi)(n):=(H_{\lambda,\theta,\rho}\psi)(n) +P(n)\psi(n),\quad \psi\in l^{2}(\Z),
\end{equation}
with the perturbation satisfying $|P(n)|\leq C(1+|n|)^{-p}$, for all~$n\in\Z$, for some~$C>0$ and~$p>3\gamma_{2}-\gamma_{1}$, is also purely~$\alpha$-Hausdorff continuous.
\end{theorem}


A particular instance of Sturmian operator is the Fibonacci operator, which corresponds to the rotation number $\theta=\frac{\sqrt{5} -1}{2}$ (the golden mean). In~\cite{DG} it is observed that, in this case with~$\lambda>0$,  
\[\gamma_{1}<\frac{\ln\left(1+\frac{1}{(2+2\lambda)^{2}}\right)}{16\ln\left(\frac{\sqrt{5} +1}{2}\right)} \ \ \ \ \ \ \ \  \textrm{and} \ \ \ \ \ \ \ \  \gamma_{2}>1+\frac{\ln[(5+2\lambda)^{1/2}(3+\lambda)c_{\lambda}]}{\ln\left(\frac{\sqrt{5} +1}{2}\right)},
\] where $c_{\lambda}$ denotes the largest root of the polynomial $x^{3}-(2+\lambda)x-1$. As an illustration, take $\lambda=1$, so that, according to Theorem~\ref{teorema1.1}, one has~$\alpha$-Hausdorff stability of the singular continuous spectrum (when it exists) under such perturbations if $p\geq 21.7.$

The result in Theorem~\ref{teorema1.1} (and  in Theorem~\ref{teorema1.2} below as well) should be contrasted with  SULE operators (see~\cite{dRJLS}), Anderson-model Hamiltonians in particular, for which rank one perturbations always result in zero-dimensional Hausdorff spectrum (point or singular continuous).

Another class of operators~\cite{JL1,tcherem} for which Hausdorff dimensional spectral properties are known is given by sparse operators~$H_{\varphi}^{\alpha}$ defined by the action~\eqref{H} in $l^{2}(\N)$, along with a phase boundary condition~\eqref{boucon} and, for each~$\alpha\in(0,1)$, sparse potentials
\begin{equation}\label{esparso2}
V(n)=\left \{
\begin{array}{cc}
x_{j}^{(1-\alpha)/2\alpha}, & n=x_{j}\in \B \\
0, & n\notin \B\\
\end{array}
\right.,
\end{equation}
where $\B=(x_{j})_{j}=\left(2^{j^{j}}\right)_{j}$. Its essential spectrum is~$[-2,2]$, and the restriction of its spectral measure to the interval $(-2,2)$ has exact Hausdorff dimension~$\alpha$ (see Definition~\ref{def.13}) for all boundary phase~$\varphi$~\cite{JL1, tcherem}. 

Again, the proof of this interesting result relies decisively on the sparseness of the potential, and here we show that it is also stable under suitable power-law decaying perturbations in case a singular continuous component is present.  More precisely, we have the following

\begin{theorem} \label{teorema1.2}
Fix~$\alpha\in(0,1)$. Let $H_{\varphi}^{{\alpha}}$ be as above and 
\begin{equation} \label {sparseP}
(H_{\varphi}^{P,{\alpha}}\psi)(n):=(H_{\varphi}^{\alpha}\psi)(n)+P(n)\psi(n),\quad \psi\in l^2(\N),
\end{equation}
with $|P(n)|\leq C(1+n)^{-p}$ for all~$n$ and some $C>0$, $p>(1+2\alpha)/\alpha$ if~$\alpha\le1/2$,~$p>(3+2\alpha)/(2\alpha)$ if~$\alpha\ge1/2$. 
Then, any possible singular continuous component of the perturbed operator~$H_{\varphi}^{P,{\alpha}}$ has also exact Hausdorff dimension~$\alpha$ for any boundary phase~$\varphi$.
\end{theorem}

To illustrate Theorem~\ref{teorema1.2}, take~$\alpha=1/2$ so that the $\alpha$-Hausdorff continuity of the singular continuous component of the spectrum of the perturbed operator is stable if~$p>4$.

The proofs of Theorems~\ref{teorema1.1} and~\ref{teorema1.2} make use of  subordinacy theory (introduced by Gilbert and Pearson in~\cite{G, GP}; see~\cite{KP} for and adaptation  to discrete operators); for this, it is necessary to control the asymptotic behavior of the solutions to the (generalized) eigenvalue equation 
\begin{equation} \label{E.A.} (H\psi)(n)=E\psi(n). \end{equation}

A solution~$\psi$ to~\eqref{E.A.} is called subordinate (at~$+\infty$) if
\[ \liminf\limits_{L\to\infty} \frac{\|\psi\|_L}{\|\Phi\|_L}=0 
\]
holds for any  solution~$\Phi$ to~\eqref{E.A.} such that~$\{\psi,\Phi\}$ is a linearly independent set; $\|\cdot\|_L$ denotes the $l^2(\N)$ norm truncated at $L>0$ ($[L]$ is the integral part of~$L$), that is,
\[
 \|\psi\|_L=\left[\sum_{n=1}^{[L]}|\psi(n)|^2+ (L-[L])|\psi([L]+1)|^2\right]^{\frac{1}{2}}.
\] In case negative values of~$n$ are meaningful, the notion of a subordinate solution at~$-\infty$ is similarly introduced. 
The standard decomposition  of a spectral measure into its pure point, singular continuous and absolutely continuous can be investigated by studying solutions to~\eqref{E.A.}.

Fix $E\in\R;$ in the following, we denote by $u_{1,\varphi,E}$ and $u_{2,\varphi,E}$ the solutions to~\eqref{E.A.} which satisfy the orthogonal initial conditions
 \begin{equation} \label{c.i.} \left\{ \begin{array}{ll}
u_{1,\varphi,E}(0)=-\sin\varphi &\ \ \  u_{2,\varphi,E}(0)=\cos\varphi
\\ u_{1,\varphi,E}(1)=\cos\varphi &\ \ \ u_{2,\varphi,E}(1)=\sin\varphi
\end{array} \right., \ \varphi \in \biggl(-\frac{\pi}{2},\frac{\pi}{2}\biggr].
\end{equation}
Note that~$u_{1,\varphi,E}$ is the solution to~\eqref{E.A.} which satisfies the boundary condition~\eqref{boucon}.

Jitomirskaya and Last have proposed~\cite{JL1,JL2} a generalization of subordinacy theory, called power-law subordinacy, which provides information about Hausdorff  dimensional properties of spectral measures (see, in particular, Theorem~1.2 in~\cite{JL1}). Namely, given~$\alpha\in(0,1]$, a solution~$\psi$ to~\eqref{E.A.} is called $\alpha$-Hausdorff subordinate (or just $\alpha$-subordinate) at $+\infty$ if
\[ \liminf\limits_{L\to\infty} \frac{\|\psi\|_L}{\|\Phi\|_{L}^{\alpha/(2-\alpha)}}=0 
\]
holds for any  solution~$\Phi$ to~\eqref{E.A.} such that~$\{\psi,\Phi\}$ is a linearly independent set. In particular, the $\alpha$-Hausdorff continuous part of the spectral measure of~$H_\varphi$ (recall that it denotes the self-adjoint realization of~$H$ with boundary condition~\eqref{boucon}) is supported on the set of energies~$E$ for which~\eqref{E.A.} does not have $\alpha$-subordinate solutions, and its  $\alpha$-Hausdorff singular part  is supported on the set of energies~$E$ for which~$u_{1,\varphi,E}$ is an $\alpha$-subordinate solution.

The proofs of Theorems~\ref{teorema1.1} and~\ref{teorema1.2} will follow from Theorem~\ref{teorema4.1} below.  In particular, we are interested in energies in the set 
\[
S(H):=\left\{E \mid  \exists \,\varphi \ \mbox{s.t.} \ u_{1,\varphi,E} \ \mbox{is a subordinate solution to~\eqref{E.A.} and}  \ u_{1,\varphi,E} \notin l^{2}(\N)
\right\}.\]  
In was found~\cite{KLS} that, for any~$\varphi$, the  singular continuous part of the spectral measure of~$H_{\varphi}$ is supported in~$S(H)$. In case of  whole-line problems, the above $S(H)$ should be replaced by~\cite{G}
\[
\left\{E \mid  \exists  \  \mbox{a solution to~\eqref{E.A.} which is subordinate at both ends} \ \pm\infty \ \mbox{and is not in}  \  l^{2}(\Z)
\right\},
\]and the singular continuous parts of the spectral measures are supported in this set; note that if no solution to~\eqref{E.A.} satisfies such condition on one end, then the corresponding energy~$E$ does not belong to the singular continuous component.

 Our general result is the following

\begin{theorem}\label{teorema4.1} Let $E\in S(H)$ and~$u_{1,\varphi,E}$, $u_{2,\varphi,E}$ be solutions to~\eqref{E.A.} satisfying the initial conditions~\eqref{c.i.}. Suppose that there exist positive  constants  $\gamma_{1}, \gamma_{2}$ such that  every solution to $(H-E)u=0$ with normalized initial conditions, i.e., $|u(0)|^2+|u(1)|^2=1$, obeys the estimates  
\begin{equation}\label{estimativa p}
 C_{1}L^{\gamma_{1}}\leq \|u\|_{L}\leq C_{2}L^{\gamma_{2}}
\end{equation}
for some $C_1(E),C_2(E)$  and all $L>0$ sufficiently large. Suppose also that, for some~$p>3\gamma_{2}-\gamma_{1}$, there exists a positive constant~$C_3$ such that, for every~$n\in\mathbb{N}$, 
\begin{equation}\label{estimativa p2}
|P(n)|\leq C_{3}(1+n)^{- p}.
\end{equation} Then, $E\in S(H+P)$, and for all $\kappa\in[0,1]$,
\begin{equation} \label{limites}
\liminf_{L\rightarrow \infty}\frac{\|u_{1,\varphi,E}\|_{L}}{\|u_{2,\varphi,E}\|_{L}^{\kappa}}=\liminf_{L\rightarrow \infty}\frac{\|v_{1,\tilde{\varphi},E}\|_{L}}{\|v_{2,\tilde{\varphi},E}\|_{L}^{\kappa}},\end{equation}
where $v_{1,\tilde{\varphi},E}$ is the solution to~\eqref{E.A.} with operator $H+P$ which satisfies the initial conditions~\eqref{c.i.} with some phase~$\tilde\varphi$, and $v_{2,\tilde{\varphi},E}$ satisfying the orthogonal  conditions (always for the operator $H+{P}$). 
\end{theorem}

We emphasize that condition~\eqref{estimativa p} in Theorem~\ref{teorema4.1} is essential and it holds true for  the operators in Theorems~\ref{teorema1.1} and~\ref{teorema1.2}; see details in Section~\ref{applic}.

Under the hypotheses of Theorem~\ref{teorema4.1}, we have a kind of stability between the sets~$S(H)$ and~$S(H+P)$; note that since the perturbation~$P$ is compact (since~$p>0$), the essential spectrum of any~$H_\varphi$ is preserved under such perturbations, so in case the singular continuous component of~$H_\varphi$ coincides with its  essential spectrum, then the singular continuous spectrum of the corresponding perturbed operators $(H+P)_{\tilde\varphi}$ will be in~$S(H)$ (in particular, any $\alpha$-Hausdorff continuous component). We again emphasize the subtlety that the spectral measure of~$(H+P)_\varphi$ of  the set $S(H+P)$ may be zero, and so Theorem~\ref{teorema4.1} gives no relevant spectral information in this case.  

Note also that, by the definition of $S(H)$, for the Sturmian model discussed in Theorem~\ref{teorema1.1}, no possible eigenvalue of~$H^{P}_{\lambda,\theta,\rho}$ belongs to the spectrum of the unperturbed operator (recall that~$\sigma_{\mathrm{p}}(H_{\lambda,\theta,\rho})=\emptyset$); however, by the preservation of the essential spectrum,  $\sigma(H_{\lambda,\theta,\rho})$ is given by the accumulation points of the possible isolated eigenvalues of finite multiplicity of~$H^{P}_{\lambda,\theta,\rho}$.  The next corollary  highlights the latter discussion.

\begin{corollary}\label{estabilidade}
Let $\theta$ be an irrational number of bounded density. Then, for every $\rho\in [0,1)$ and $\lambda\neq 0,$ no perturbation of the form~\eqref{estimativa p2}  has eigenvalues in~$\sigma(H_{\lambda,\theta,\rho})$.
\end{corollary}

Obviously, by Theorem~\ref{teorema1.2}, we have an analogous version of Corollary~\ref{estabilidade} to the class of sparse operators $H_{\varphi}^{P,{\alpha}}$ of type~\eqref{sparseP}.

The organization of this paper is as follows. In Section~\ref{hmds}, we recall definitions and properties of Hausdorff measures and dimensions, as well as their role in subordinacy theory. In Section~\ref{perturb}, we prove our main general result, that is, Theorem~\ref{teorema4.1}. In Section~\ref{applic}, we prove Theorems~\ref{teorema1.1} and~\ref{teorema1.2} as direct consequences of Theorem~\ref{teorema4.1} and known results in the literature.

\section{Hausdorff measures, dimensions and subordinacy theory}\label{hmds}
We recall in this section some concepts and results regarding Hausdorff measures and subordinacy theory. Most of the material exposed here is based on \cite{F, JL1, L, Ro}.

\begin{definition}\label{def.11} Given a set $S\subset \R$ and $\alpha\in[0,1],$ consider the number 
\[
Q_{\alpha,\delta}(S)=\inf\left\{\sum_{k=1}^{\infty}|I_{k}|^{\alpha} \ \mid \ |I_{k}|<\delta, \ \forall k; \  S\subset \bigcup_{k=1}^{\infty}I_{k}\right\},
\]
with the infimum taken over all covers of~$S$ by intervals $\{I_{k}\}_k$ of size at most~$\delta$. The limit
$$h^{\alpha}(S)=\disp\lim_{\delta \to 0}Q_{\alpha,\delta}(S)$$
is called the~$\alpha$-dimensional Hausdorff measure of $S$.
\end{definition}

\begin{remak} The counting measure (which assigns to each set $S$ the number of points in it), for $\alpha=0$, and the Lebesgue measure, for $\alpha=1$, are particular cases of~$h^\alpha$.
\end{remak}

The~$\alpha$-dimensional Hausdorff measure,~$h^\alpha$, is an outer measure on subsets of $\R$~\cite{Ro}. It is known that for every set~$S$, there is a unique $\alpha_{S}$ such that $h^{\alpha}(S)=0$ if $\alpha>\alpha_{S}$ and $h^{\alpha}(S)=\infty$ if $\alpha_{S}>\alpha$. The number $\alpha_{S}$ is called the Hausdorff dimension of the set~$S$, usually denoted by $\dim_{H}(S)$.

 Now we recall some notions of continuity and singularity of Borel measures with respect to Hausdorff measures and dimensions. 

\begin{definition}\label{def.12} \
Let~$\mu$ be a Borel measure in~$\R$ and $\alpha\in[0,1]$.
\begin{description}
\item{\bf (i)}~$\mu$ is called $\alpha$-Hausdorff continuous if $\mu(S)=0$ for every Borel set $S$ with $h^{\alpha}(S)=0$.
\item{\bf (ii)}~$\mu$ is called $\alpha$-Hausdorff singular if~$\mu$ is supported on some Borel set $S$, i.e., $\mu(\R\backslash S)=0$ with $h^{\alpha}(S)=0$.
\end{description}
\end{definition}

\begin{definition}\label{def.13} \ A Borel measure~$\mu$ in~$\R$ is said to have exact Hausdorff dimension~$\alpha$, for some $\alpha\in(0,1)$, and denoted by $\dim_{H}(\mu)$, if two requirements hold: \begin{description} \item{\bf (i)} for every set $S$ with $\dim_{H}(S)<\alpha$, one has $\mu(S)=0$; \item{\bf (ii)} there is a Borel set,~$S_{0}$, of Hausdorff dimension~$\alpha$ which supports~$\mu$. \end{description} \end{definition}

A  Borel measure~$\mu$ in~$\R$ is said to be zero-Hausdorff dimensional if it is supported on a set with $\dim_{H}(S)=0$, and, for $\mu\ne0$, one-Hausdorff dimensional if $\mu(S)=0$ for any set $S$ with $\dim_{H}(S)<1$.
  
\begin{remak} According to Definitions~\ref{def.12} and~\ref{def.13}, a Borel measure~$\mu$ in~$\R$ is of exact Hausdorff dimension~$\alpha$ if, for every $\varepsilon>0$, it is simultaneously $(\alpha-\varepsilon)$-continuous and $(\alpha+\varepsilon)$-singular. \end{remak}

Given a finite Borel measure~$\mu$ and $\alpha\in[0,1]$, define
\[ 
D^{\alpha}_{\mu}(E):=\limsup_{\varepsilon \to 0}\frac{\mu((E-\varepsilon,E+\varepsilon))}{(2\varepsilon)^{\alpha}} 
\]
and set~$T^{\alpha}_{\infty}=\{E\in\R\mid D^{\alpha}_{\mu}(E)=\infty\}$ (which is a Borelian). The restriction $\mu_{\alpha s}:=\mu(T^{\alpha}_{\infty}\cap\cdot)$ is $\alpha$-Hausdorff singular, and $\mu_{\alpha c}:=\mu((\R\backslash T^{\alpha}_{\infty})\cap\cdot)$ is $\alpha$-Hausdorff continuous. Thus, each finite Borel measure decomposes uniquely into an $\alpha$-Hausdorff continuous part and an $\alpha$-Hausdorff singular part: $\mu=\mu_{\alpha s}+\mu_{\alpha c}$. Moreover, an $\alpha$-Hausdorff singular measure is such that $D_{\mu}^{\alpha}(E)=\infty$ a.e (with respect to it), while an $\alpha$-Hausdorff continuous measure is such that $D_{\mu}^{\alpha}(E)<\infty$ a.e (see Chapter~3 in~\cite{Ro}).

The result in~\cite{JL1} that connects  Hausdorff singularity and continuity of the spectral measure of~$H$ with the scaling behavior of the solutions to~\eqref{E.A.} is the following:

\begin{theorem} [Theorem~$1.2$ in~\cite{JL1}] \label{teo1.2JL1}
Let~$H_\varphi$ be defined by~\eqref{H}-\eqref{boucon} in $l^{2}(\N)$, and~$\mu$ denote the spectral measure of~$H_\varphi$ associated with the cyclic vector $\delta_{1}$. Let~$E\in\R$ and $\alpha \in (0,1)$. Then, for any $\varphi\in (-\pi/2,\pi/2]$,
\[D^{\alpha}_{\mu}(E)=\infty\]
holds if, and only if, $u_{1,\varphi,E}$ is $\alpha$-subordinate, that is,
 \[\liminf_{L \to \infty}
\frac{\|u_{1,\varphi,E}\|_{L}}
{\|u_{2,\varphi,E}\|^{\alpha/(2-\alpha)}_{L}}=0.\]
\end{theorem}

 Theorem~\ref{teo1.2JL1} provides an effective tool for the analysis of Hausdorff dimensional properties of spectral measures of Schr\"odinger operators. Namely, the $\alpha$-Hausdorff continuous part of~$\mu$ is supported on the set of energies~$E$ for which~\eqref{E.A.} does not have $\alpha$-subordinate solutions, and the $\alpha$-Hausdorff singular part of~$\mu$ is supported on the set of energies~$E$ for which~$u_{1,\varphi,E}$ is $\alpha$-subordinate.

Consequently, one may use Theorem~\ref{teorema4.1} to obtain information about the Hausdorff dimension of spectral measure of the Schr\"odinger operator studied, and this is an essential tool in the proof of Theorems~\ref{teorema1.1} and~\ref{teorema1.2}.

\section{A general result}\label{perturb}

We present in this section the proof of Theorem~\ref{teorema4.1}, which is based on results in~\cite{KLS}. Suppose that the behavior of the solutions to the eigenvalue equation~\eqref{E.A.}  for~$V = V_ {0}$ is known;  the idea is to use this knowledge to determine the behavior of the solutions to~\eqref{E.A.} for the potential~$V = V_ {0} + P$, with the perturbation~$P$ decaying as in~\eqref{estimativa p2}.

In order to avoid cumbersome notations, we set~$u_{1}:=u_{1,\varphi,E}$, the subordinate solution for~$V=V_{0}$, and $u_{2}:=u_{2,\varphi,E}$ the corresponding solution satisfying the orthogonal initial conditions~\eqref{c.i.}. As usual~\cite{KLS}, we apply the variation of parameters method in order to obtain a linearly independent system of solutions to~\eqref{E.A.} for~$V = V_ {0} + P$; namely, we will look for solutions~$v$ in the form
\[v(n)=w_{1}(n)u_{1}(n)+w_{2}(n)u_{2}(n)
\]
and such that
\[v(n-1)-v(n)=w_{1}(n)\left[u_{1}(n-1)-u_1(n)\right]+ w_2(n)\left[u_{2}(n-1)-u_2(n)\right].
\]
By writing~$w(n):=\left(\begin{array}{c}w_{1}(n)\\
 w_{2}(n)\end{array}\right)$, the eigenvalue equation~\eqref{E.A.} for~$V = V_ {0} + P$ is equivalent to
\begin{equation}\label{eq: vp}
w(n+1)-w(n)=A(n)w(n),\end{equation}
with
\begin{equation*} 
A(n)=-P(n)\left(\begin{array}{ll}
u_{1}(n)u_{2}(n) \ & \   u_{2}(n)^{2}
\\ -u_{1}(n)^{2} \ & \  -u_{1}(n)u_{2}(n)  \end{array}
\right).\end{equation*} 

For a positive monotone increasing function $f:\{0,1,2,\cdots\}\to (0,\infty)$, let
\[
G(n):=\max \left\{|P(n)|\left(|u_{1}(n)u_{2}(n)|+ |u_{2}(n)|^{2}\right); |P(n)|\left(f(n)|u_{1}(n)|^{2}+ |u_{1}(n)u_{2}(n)| \right)\right\}.
\]
\begin{lemma} \label{teorema2.2}
Let $f$ be as above and suppose that  
\[\sum_{n=1}^{\infty}G(n)<\infty.\]
Then, there exist solutions $w^{\pm}$ to~\eqref{eq: vp} such that,  as~$n\rightarrow\infty$,
\begin{description}
	\item{\bf (i)} $w_{1}^{-}(n)\rightarrow 1$ \ \ and \ \  $f(n)w^{-}_{2}(n)\rightarrow 0$,
	\item{\bf (ii)} $w_{1}^{+}(n)\rightarrow 0$ \ \ and \ \   $w^{+}_{2}(n)\rightarrow 1$.
\end{description}
\end{lemma}
\begin{proof}
The proof of Lemma~\ref{teorema2.2} traces the same steps of the proof of Theorem~2.2 in~\cite{KLS}, with obvious adaptations for the discrete case.
\end{proof}

 In the proof of Theorem~\ref{teorema4.1}  we will need to choose a function~$f$ so that~$G\in l^{1}(\N)$  and also to connect with  ideas of Jitomirskaya and Last~\cite{JL1}. So, the following result will be useful, which is a (not completely immediate) discrete version of Lemma~3.2 in~\cite{KLS}. 

\begin{lemma} \label{lema4.3}
Let $\left\{\xi(n)\right\}$ be a sequence of numbers such that, for some positive constant~$C_1$, 
\begin{equation}\label{eq: desig4.1}
|\xi(n)|\leq C_{1}(1+n)^{-a},
\end{equation}
and let~$\psi_{1},\psi_{2}$ be solutions to~\eqref{E.A.} satisfying
\begin{equation}\label{eq: desig4.2}
\|\psi_{1}\|_{L}\|\psi_{2}\|_{L}\leq C_{2}(1+L)^{b},
\end{equation} for each~$L\in\mathbb{N}$. If~$a>b>0$, then
\[\sum_{n=1}^{\infty}|\xi(n)\psi_{1}(n)\psi_{2}(n)|<\infty.\]
\end{lemma}

\begin{proof}
Let~$g:\{0,1,\ldots\}\rightarrow\mathbb{R}$ be given by~$g(n):=\sum_{j=1}^{n}|\psi_{1}(j)\psi_{2}(j)|$ for~$n\ge 1$ and $g(0)=0$. By Cauchy-Schwarz inequality and~\eqref{eq: desig4.2},
\begin{equation}\label{eq: desig4.3}
g(n)\leq C_{2}(1+n)^{b}.
\end{equation}
Without loss, we simplify by taking $C_1=C_2=1$. By~\eqref{eq: desig4.1}, for each~$L\in\mathbb{N}$,
\begin{eqnarray*}
\sum_{n=1}^{L}|\xi(n)\psi_{1}(n)\psi_{2}(n)|&\leq&\sum_{n=1}^{L}(1+n)^{-a}|\psi_{1}(n)\psi_{2}(n)|\\
&=&\sum_{n=1}^{L}(1+n)^{-a}[g(n)-g(n-1)]\\
&=& (2+L)^{-a}g(L) + \sum_{n=1}^{L}[(1+n)^{-a}- (2+n)^{-a}]g(n)\\
&\le& (2+L)^{-a}g(L) + \sum_{n=1}^{L}a(1+n)^{-a-1}g(n);
\end{eqnarray*}
the second inequality is a consequence of the Mean Value Theorem applied to the function $h(x)=(1+x)^{-a}$, $x\ge0$, and the inequality
\[\max_{n\leq x\leq n+1}|h'(x)|\leq a(1+n)^{-a-1}.\]
Therefore, by~\eqref{eq: desig4.3}, one has
\[
\sum_{n=1}^{L}|\xi(n)\psi_{1}(n)\psi_{2}(n)|\leq (2+L)^{-a}L^{b} + a\sum_{n=1}^{L}(1+n)^{b-a-1}.
\]
Now, since~$a>b$, it follows that
\[\lim_{L\rightarrow\infty}\sum_{n=1}^{L}|\xi(n)\psi_{1}(n)\psi_{2}(n)|\leq a\sum_{n=1}^{\infty}(1+n)^{b-a-1}<\infty.\]
\end{proof}

By the definition of~$G(n)$, in order to show that $\sum_{n=1}^{^\infty}G(n)<\infty$, it  is sufficient to show that each of the series
\[
\sum_{n=1}^{\infty}|P(n)u_{1}(n)u_{2}(n)|, \ \ \sum_{n=1}^{\infty}|P(n)||u_{2}(n)|^{2}, \ \ \sum_{n=1}^{\infty}f(n)|P(n)||u_{1}(n)|^{2}\]
is finite. 

In what follows we set, for each~$\gamma>0$, $f(n)=(1+n)^\gamma$. Hence, by Lemma~\ref{lema4.3} and the hypotheses~\eqref{estimativa p} and~\eqref{estimativa p2}, the above series are finite if
\begin{equation} \label{eq: gama}
p>2\gamma_{2} \ \ \  \textrm{and} \ \ \  p-\gamma>2\gamma_{2};
\end{equation}
namely, we can choose~$\gamma>0$ so that $\gamma_{2}-\gamma_{1}<\gamma<p-2\gamma_{2}$ (recall $p>3\gamma_{2}-\gamma_{1}$), and consequently~\eqref{eq: gama} holds true, implying that~$G\in l^1(\N)$.

\begin{proof} {\it (Theorem~\ref{teorema4.1})}
By the  above considerations, there exist solutions~$w^{\pm}$ to~\eqref{eq: vp}, given in Lemma~\ref{teorema2.2}, so that
\[v_1(n)=w_{1}^{-}(n)u_{1}(n)+w_{2}^{-}(n)u_{2}(n)
\]
\[v_{2}(n)=w_{1}^{+}(n)u_{1}(n)+w_{2}^{+}(n)u_{2}(n)
\]
are linearly independent solutions to~$(H+P)v=Ev$. Thus, it suffices to prove that
\begin{equation} \label{converge}
\frac{\|v_{1}\|_{L}}{\|u_{1}\|_{L}}\longrightarrow 1 \quad \textrm{and} \quad \frac{\|v_{2}\|_{L}}{\|u_{2}\|_{L}}\longrightarrow 1,
\end{equation}
as $L\rightarrow\infty$. Namely, by~\eqref{converge}, one has, for each~$\kappa\in[0,1]$, 
\[
{\left[\frac{\|v_{1}\|_{L}}{\|v_{2}\|_{L}^{\kappa}}\right]}{\left[\frac{\|u_{1}\|_{L}}{\|u_{2}\|_{L}^{\kappa}}\right]^{-1}}\underset{L\rightarrow\infty}{\longrightarrow}1,
\]
which proves  assertion~\eqref{limites}. Note that~\eqref{converge} also ensures that~$v_1$ is not square summable.

In order to prove~\eqref{converge}, we begin observing that
\begin{eqnarray*}
\frac{|\|v_{1}\|_{L}- \|u_{1}\|_{L}|}{\|u_{1}\|_{L}}&\leq& \frac{\|v_{1}-u_{1}\|_{L}}{\|u_{1}\|_{L}}\\
&\leq& \frac{\|v_{1}-w_{1}^{-}u_{1}\|_{L}}{\|u_{1}\|_{L}} + \frac{\|(w_{1}^{-}-1)u_{1}\|_{L}}{\|u_{1}\|_{L}}\\
&=&\frac{\|w_{2}^{-}u_{2}\|_{L}}{\|u_{1}\|_{L}} + \frac{\|(w_{1}^{-}-1)u_{1}\|_{L}}{\|u_{1}\|_{L}}.
\end{eqnarray*}
Since $w_{1}^{-}(n)\rightarrow 1$, then for every $\varepsilon>0$, there exists an integer~$n_{0}$ such that $|w_{1}^{-}(n)-1|<\varepsilon$, for every~$n\geq n_{0}$. Hence, for each integer~$L>n_0$,
\[
\frac{\|(w_{1}^{-}-1)u_{1}\|^{2}_{L}}{\|u_{1}\|^{2}_{L}}\leq \frac{\sum_{n=1}^{n_{0}}|(w_{1}^{-}(n)-1)u_{1}(n)|^{2}}{\|u_{1}\|^{2}_{L}}+\varepsilon^{2},
\]
and consequently, 
\[\frac{\|(w_{1}^{-}-1)u_{1}\|_{L}}{\|u_{1}\|_{L}}\underset{L\rightarrow\infty}{\longrightarrow}0.
\]
One also has that $f(L)w^{-}_{2}(L)\rightarrow 0$ and $f(L)=(1+L)^{\gamma}$, with $\gamma>\gamma_2-\gamma_1$. Thus, there exists a positive constant~$C$ such that 
\[\frac{\|w_{2}^{-}u_{2}\|^{2}_{L}}{\|u_{1}\|^{2}_{L}}\leq \frac{C\sum_{n=1}^{L}(1+n)^{-2\gamma}|u_{2}(n)|^{2}}{\|u_{1}\|^{2}_{L}}.
\]
As in Lemma~\ref{lema4.3}, one can write (again with $C_1=C_2=1$)
\[
\sum_{n=1}^{L}(1+n)^{-2\gamma}|u_{2}(n)|^{2}\leq (2+L)^{2\gamma_{2}-2\gamma}+2\gamma \sum_{n=1}^{L}(1+n)^{2\gamma_{2}-2\gamma -1},
\]
and since $\|u_{1}\|^{2}_{L}\geq L^{2\gamma_{1}}$, one has
\[\frac{\|w_{2}^{-}u_{2}\|_{L}}{\|u_{1}\|_{L}}\underset{L\rightarrow\infty}{\longrightarrow}0;
\] therefore, 
\[\frac{\|v_{1}\|_{L}}{\|u_{1}\|_{L}}\underset{L\rightarrow\infty}{\longrightarrow}1.
\]

Similarly to what has been presented above, one has
\begin{eqnarray*}
\frac{|\|v_{2}\|_{L}- \|u_{2}\|_{L}|}{\|u_{2}\|_{L}}&\leq& \frac{\|v_{2}-u_{2}\|_{L}}{\|u_{2}\|_{L}}\\
&\leq& \frac{\|v_{2}-w_{2}^{+}u_{2}\|_{L}}{\|u_{2}\|_{L}} + \frac{\|(w_{2}^{+}-1)u_{2}\|_{L}}{\|u_{2}\|_{L}}\\
&=&\frac{\|w_{1}^{+}u_{1}\|_{L}}{\|u_{2}\|_{L}} + \frac{\|(w_{2}^{+}-1)u_{2}\|_{L}}{\|u_{2}\|_{L}}.
\end{eqnarray*} 
Now, since $w_{2}^{+}(n)\rightarrow 1$ and  $w_{1}^{+}(n)\rightarrow 0$, as $n\rightarrow\infty$, and since~$u_{1}$ is a subordinate solution, it follows that both terms on the right-hand side of the above inequality  tend to zero as~$L\rightarrow\infty$; hence,
\[\frac{\|v_{2}\|_{L}}{\|u_{2}\|_{L}}\underset{L\rightarrow\infty}{\longrightarrow}1,
\]as required.

\end{proof}

\section{Applications: Proofs of Theorems~\ref{teorema1.1} and~\ref{teorema1.2}}\label{applic}

In this section we use Theorem~\ref{teorema4.1} to conclude Theorems~\ref{teorema1.1} and~\ref{teorema1.2}, whose proofs are now rather easy. We begin with~$H^{P}_{\lambda,\theta,\rho}$ given by~\eqref{P.S.}. Recall that given an irrational~$\theta \in [0,1)$, it has an infinite continued fraction expansion~\cite{Khi}
\[ \theta = \dfrac{1}{a_1 + \dfrac{1}{a_2 + \dfrac{1}{a_3 +\cdots}}} \]
with uniquely determined $a_{n}\in \N$. The number $\theta$ is said to have bounded density if
\[ \limsup_{n\to \infty}\frac{1}{n}\sum_{i=1}^{n}a_i<\infty. \]

\begin{proof} {\it (Theorem~\ref{teorema1.1})}
It is  known \cite{IRT,D,DKL} that, for Schr\"odinger operators with Sturmian potentials whose rotation number is of bounded density, there exist power-law bounds of the form
\[C_1L^{\gamma_1}\leq \|u\|_L \leq C_2L^{\gamma_2}\]
for every solution $u$ to~\eqref{E.A.} (with normalized initial conditions~\eqref{NIC}); with these estimates, it is possible to prove the nonexistence of $\alpha$-subordinate solutions for $\alpha = \frac{2\gamma_1}{\gamma_1 + \gamma_2}$.
More specifically, it was proven~\cite{D,DKL,JL2} that if~$\theta$ has bounded density, then for every $\lambda\neq 0,$ there exists $\alpha=\alpha(\lambda,\theta)>0$ such that for every $\rho\in [0,1)$, the spectral measure of $H_{\lambda,\theta,\rho}$ is purely $\alpha$-Hausdorff continuous. 

We note~\cite{D,DKL} that if one is able to establish uniform power-law bounds on the restriction of the operator to the right half-line, then the resulting~$\alpha$-continuity is independent of the potential on the left half-line. In this sense, the more continuous half-line dominates and bounds the dimensionality of the whole-line problem from below.  

 Suppose that $\sigma(H^{P}_{\lambda,\theta,\rho})$ has some singular continuous component; now, since the perturbation decays as~$|P(n)|\leq C(1+|n|)^{-p}$, with $p>3\gamma_{2}-\gamma_{1}$, it is a compact perturbation and the essential spectrum is preserved. Thus,~$S(H_{\lambda,\theta,\rho})\supset\sigma_{\mathrm{sc}}(H^{P}_{\lambda,\theta,\rho})$, and by Theorem~\ref{teorema4.1}, we obtain that the asymptotic behavior of generalized eigenfunctions of the operators~$H^{P}_{\lambda,\theta,\rho}$ (that is, the solutions to~\eqref{E.A.})  in~\eqref{P.S.} is analogous to the behavior of eigenfunctions of the operators~$H_{\lambda,\theta,\rho}$;  and again by the $\alpha$-subordinacy theory, such component is still $\alpha$-Hausdorff continuous for these perturbed operators, with $\alpha = \frac{2\gamma_1}{\gamma_1 + \gamma_2}$.
\end{proof}

\

\begin{proof} {\it (Theorem~\ref{teorema1.2})}
In Theorem~$1.3$ in~\cite{JL1}, it was shown that the spectral measure of the operator $H_{\varphi}^{\alpha}$ restricted to $(-2,2)$, with potential $V_{0}=V$ given by~\eqref{esparso2}, has exact Hausdorff dimension~$\alpha$ for (Lebesgue) a.e.\ boundary phase $\varphi\in(-\pi/2,\pi/2]$. However, Tcheremchantsev presented in~\cite{tcherem} (item~2 of Corollary~4.5) an improvement of this result, showing that this spectral measure restricted to $(-2,2)$ has, in fact, exact Hausdorff dimension~$\alpha$ for any boundary phase~$\varphi\in(-\pi/2,\pi/2]$.

It follows from inequalities~(5.6) and~(5.7) in~\cite{JL1} that, for sufficiently large~$L$, the estimate~\eqref{estimativa p} is satisfied with $\gamma_{2}>(1+\alpha)/(2\alpha)$ for each~$\alpha\in(0,1)$, and~
\begin{eqnarray*}
\gamma_1<\left\{\begin{array}{ll}(1-\alpha)/\alpha&\mathrm{if}\;\;\alpha\le1/2,\\
1/2&\mathrm{otherwise}\end{array}\right. .
\end{eqnarray*} Therefore, as in the proof of Theorem~\ref{teorema1.1}, by a direct consequence of Theorem~\ref{teorema4.1}, we obtain that the asymptotic behavior of generalized eigenfunctions of the operators~$H_{\varphi}^{P,{\alpha}}$ in~\eqref{P.S.} is similar to the behavior of eigenfunctions of the operators~$H_{\varphi}^{\alpha}$;  and again by  power-law subordinacy theory, it follows that any possible singular continuous component of the restriction of the spectral measure of the operator~$H_{\varphi}^{P,{\alpha}}$ to~$(-2,2)$ has also exact Hausdorff dimension~$\alpha$ for any boundary phase~$\varphi\in(-\pi/2,\pi/2]$.
\end{proof}

 \
 
 \noindent Acknowledgment: VRB thanks  CAPES (a Brazilian agency) for financial support. CRdO  thanks the partial support of CNPq (Universal Project 41004/2014-8).
 
 \

\end{document}